\documentclass[letterpaper,11pt]{article}
\usepackage[margin=0.97in]{geometry}
\usepackage{amsthm,amsmath}
\usepackage{amssymb}
\usepackage{times}

\usepackage[ruled]{algorithm}
\usepackage[noend]{algpseudocode}

\newif\iffull
\fullfalse

\usepackage{epsfig}

\renewcommand{\thefootnote}{\fnsymbol{footnote}}

\newtheorem{theorem}{Theorem}
\newtheorem{corollary}{Corollary}
\newtheorem{lemma}{Lemma}

\newtheorem{proposition}{Proposition}
\newtheorem{fact}{Fact}

\newtheoremstyle{redstyle}
     {3pt}
     {3pt}
     {\color{black}}
     {}
     {\color{red}\bfseries}
     {:}
     {.5em}
     {}
\theoremstyle{redstyle}

\newcommand{\m}{\mathcal}
\newcommand{\cT}{{\mathcal T}}
\newcommand{\cN}{{\mathcal N}}

\newcommand{\cI}{{\mathcal I}}
\newcommand{\ep}{\varepsilon}
\newcommand{\eps}{\varepsilon}

\newcommand{\DIR}{\text{DIR}}

\newcommand{\dogory}{\vspace*{-3pt}}

\newcommand{\labell}[1]{\label{#1}}
\renewcommand{\paragraph}[1]{\vspace*{0.5ex}\noindent {\bf #1}}
\newcommand{\remove}[1]{}

\newcommand{\dist}{\text{dist}}
\newcommand{\distM}{\text{distM}}

\newcommand{\NAT}{{\mathbb N}}
\newcommand{\INT}{{\mathbb Z}}

\usepackage{color}
\newcommand{\dk}[1]{#1}
\newcommand{\tj}[1]{#1} 
\newcommand{\comment}[1]{}

\newcommand{\boxx}{\text{box}}

\begin{document}

\title{Distributed Deterministic Broadcasting\\ in Uniform-Power Ad Hoc Wireless Networks}

\author{%
    Tomasz Jurdzinski\footnotemark[2]
    \and
    Dariusz R.~Kowalski\footnotemark[3]
    \and
    Grzegorz Stachowiak\footnotemark[2]
}

\footnotetext[2]{Institute of Computer Science, University of Wroc{\l}aw, Poland.}

\footnotetext[3]{Department of Computer Science,
            University of Liverpool,
            Liverpool L69 3BX, UK.
            }


\date{}

\maketitle


\begin{abstract}
Development of many futuristic technologies, such as MANET, VANET, iThings, nano-devices,
depend on efficient distributed communication protocols in multi-hop ad hoc networks.
A vast majority of research in this area focus on design heuristic protocols, and
analyze their performance by simulations on networks generated randomly or obtained
in practical measurements of some (usually small-size) wireless networks.
Moreover, they often assume access to truly random sources, which is often not reasonable in case of
wireless devices.
In this work we use a formal framework to study the problem of broadcasting and its time complexity in
{\em any} two dimensional Euclidean wireless network with uniform transmission powers.
For the analysis, we consider two popular models of ad hoc networks
based on the Signal-to-Interference-and-Noise Ratio (SINR): one with opportunistic links, and
the other with randomly disturbed SINR.
In the former model, we show that one of our algorithms accomplishes broadcasting in $O(D\log^2 n)$
rounds, where $n$ is the number of nodes and $D$ is the diameter of the network.
If nodes know a priori the granularity $g$ of the network, i.e., the inverse of the maximum transmission range
over the minimum distance between any two stations, a modification of this algorithm
accomplishes broadcasting in $O(D\log g)$ rounds.
Finally, we modify both algorithms to make them efficient in the latter model with randomly disturbed SINR,
with only logarithmic growth of performance.
Ours are the first provably efficient and well-scalable, under the two models,
distributed deterministic solutions for the broadcast task.

\vspace*{1ex}

\end{abstract}
\ \\

\thispagestyle{empty}
\setcounter{page}{0}


\renewcommand{\thefootnote}{\arabic{footnote}}

\newpage

\section{Introduction}

In this work we consider a broadcast problem in ad-hoc wireless networks under the
Signal-to-Interference-and-Noise-Ratio (SINR) models.
Wireless network consists of at most $n$ stations, also called nodes, with unique integer IDs
and uniform transmission powers $P$, deployed in the two-dimensional space with Euclidean metric.
Each station initially knows only its own ID, location and
the upper bound $n$ on the number of nodes.

There are three dominating approaches in the literature to analyze performance of
communication protocols
in wireless networks based on the Signal-to-Interference-and-Noise-Ratio.
The first model is based on the {\em SINR with random disturbances}: each measured SINR
is randomly changed according to some stochastic distribution~\cite{losowy-szum}.
This model takes into account all signal disturbances that may occur in the environment,
apart of the average background noise and signal deterioration captured by the
deterministic SINR formula.

The second model is based on the notion of communication graph, which contains
only those communication links that are between stations of distance $(1-\ep)$
times the maximum transmission range. Even though communication between
two nodes not directly connected in such communication graph may occur in practice,
it is very unlikely especially if relatively many nodes intend to transmit.
Therefore, intuitively, the algorithm may rely only on the links in the communication graph,
and in order to be fairly treated, its performance should be analyzed as if only those links
were available. For example, the diameter of the communication graph is a natural
lower bound on broadcasting problem, even though in practice there might exist
shorter paths (but being very difficult to propagate any message along themselves, due to
large distances and substantial signal deterioration). We call this setting
the {\em model with opportunistic links}, c.f.,~\cite{YuHWTL12}.

The third model is based on additional assumption that in order to receive a message,
not only the SINR must be above some (relatively high) threshold, but also
the received dominating signal must be sufficiently large.
We call this setting
the {\em SINR model with weak devices}, c.f.,~\cite{JKS-arxiv,KV10}.

%
%
%
In this work we focus only on the first two models.
We consider two settings: one with no local knowledge being provided a priori to the nodes,
and the other where each node knows  the granularity $g$ of the network,
i.e., the inverse of the maximum transmission range
over the minimum distance between any two stations.

In the broadcast problem considered in this work, there is one designated node, called a source, which has a piece
of information, called a source message or a broadcast message,
which must be delivered to all other accessible (not necessarily directly)
nodes by using wireless communication.
In the beginning, only the source is executing the broadcast protocol,
and the other nodes join the execution after receiving the broadcast message for the first time.
The goal is to minimize time needed for accomplishing the broadcast~task.

\vspace*{-1ex}
\subsection{Previous and Related Results}

In this work, we study
the performance of {\em distributed deterministic broadcasting} in ad hoc wireless networks
under the two SINR-based physical models mentioned above.
%
%
In what follows, we discuss most relevant results in the SINR-based models, and
the state of the art obtained in the older Radio Network model.

\paragraph{SINR models.}
In the 
SINR model \dk{with opportunistic links},
slightly weaker task of {\em local} broadcasting in ad hoc setting, in which nodes have to inform
only their neighbors in the corresponding communication graph,
was studied in \cite{YuWHL11}.
The considered setting allowed power control by deterministic algorithms,
in which, in order to avoid collisions,
stations could transmit with any power smaller than the maximal one.
Randomized solutions for contention resolution~\cite{KV10}
and local broadcasting~\cite{GoussevskaiaMW08} were also obtained.
Recently, a distributed randomized algorithm for multi-broadcast has been
presented \cite{YuHWTL12} for uniform networks. Although the problem solved in that paper
is a generalization
of broadcast, the presented solution is restricted merely to networks having
the communication graph connected for $\eps=\frac23 r$, where $r$
is the largest possible SINR ratio. In contrast, our solutions are efficient and scalable
for {\em any} networks with communication graph connected for {\em any} value of $\eps<\frac{1}{2}$.
\dk{(In case of $\ep\in [1/2,1)$, one could take our algorithm for $\ep'=1/3$,
which guarantees at least as good asymptotic performance.)}


\dk{In the SINR model with random disturbances, motivated by many practical works, c.f.,~\cite{losowy-szum}, we are not aware of any theoretical analysis
of distributed deterministic broadcasting problem.}

\dk{In the model of weak devices, c.f.,~\cite{KV10} studying randomized local broadcast,
broadcasting algorithms are not
scalable (in terms of network diameter), unless nodes know their neighbors in the
corresponding communication graph in advance~\cite{JKS-arxiv}. This is a fundamental difference between
the models considered in this paper, which do not impose any additional physical constraints
on receiving devices apart of the SINR threshold, and the model of weak devices which cannot
decode weak signals.
On the positive side, a scalable (in terms of the maximum node degree) distributed deterministic
construction of efficient backbone sub-network was showed in~\cite{JK-disc12}.
Once such a network is spanned, scalable (in terms of the diameter of the original network)
distributed
solutions to many communication tasks can be constructed.
}

There is a vast amount of work on centralized algorithms under the classical SINR models.
The most studied problems include connectivity, capacity maximization,
link scheduling types of problems;
for recent results and references we refer the reader to the survey~\cite{WatSurv}.

\paragraph{Radio network model.}
There are several papers analyzing broadcasting in the radio model of wireless networks,
under which a message is successfully heard if there are no other simultaneous transmissions
from the {\em neighbors} of the receiver in the communication graph.
This model does not take into account the real strength of the received signals, and also the signals
from outside of some close proximity.
In the {\em geometric} ad hoc setting, Dessmark and Pelc~\cite{DessmarkP07} were the first who studied
the broadcast problem.
They analyzed the impact of local knowledge, defined as a range within which
stations can discover the nearby stations.
Emek et al.~\cite{EmekGKPPS09} designed a broadcast algorithm
working in time $O(Dg)$
in UDG radio networks with eccentricity $D$ and granularity $g$, where
eccentricity was defined as the minimum number of hops to propagate the broadcast message throughout
the whole network.
Later, Emek et al.~\cite{EmekKP08} developed a matching lower bound $\Omega(Dg)$.
Mobility aspects of communication were studied in \cite{Farach-ColtonM07}.
There were several works analyzing deterministic broadcasting in geometric graphs in the centralized radio setting,
c.f.,~\cite{GasieniecKLW08,SenH96}.

The problem of broadcasting is well-studied in the setting of {\em graph radio model}, in which stations
are not necessarily deployed in a metric space;
here we restrict to only the most relevant results.
In deterministic ad hoc setting with no local knowledge, the fastest $O(n\log(n/D))$-time algorithm in symmetric networks was developed by Kowalski~\cite{Kow-PODC-05}, and almost matching lower
bound was given by Kowalski and Pelc~\cite{KP-DC-05}.
For recent results and references in less related settings we refer the reader
to~\cite{DeMarco-SICOMP-10, KP-DC-05,Censor-HillelGKLN11}
%
There is also a vast literature on randomized algorithms for broadcasting in graph radio model
\cite{KushilevitzM98,KP-DC-05,CzumajRytter-FOCS-03}.


\vspace*{-1.2ex}
\subsection{Our Results}

\dk{
In this paper we present distributed deterministic algorithms for broadcasting in
ad hoc wireless networks of uniform transmission power, deployed in two-dimensional Euclidean space.
The time performance of these protocols is measured in two SINR-based models:
with opportunistic links and with random disturbances.
%

In the former model,
when no knowledge of the network topology is a priori provided to the nodes,
except of the upper bound $n$ on the number of nodes,
one of our algorithms works in $O(D\log^2 n)$ rounds.
%
A variation of this protocol accomplishes broadcast in
$O(D\log g)$ rounds
in case when nodes know the
network granularity before the computation.
(It is sufficient that only the source a priori knows network granularity.)}

We show that our algorithms can be easily transformed to achieve similar performance,
bigger by factor $O(\log n)$, in the latter model, with high probability
(i.e., with probability at least $1-n^{-c}$, for some suitable constant $c>1$).
Another useful property that could be almost immediately derived from this transformation
is that nodes do not need to know their exact positions, but only their estimates ---
this inaccuracy could be overcome by setting a slightly smaller deviation parameter $\eta$
of the stochastic distribution of random disturbances (although this may in turn result in
increasing the error probability $\zeta$ of deviating SINR by factor outside of the range
$(1-\eta,1+\eta)$, the asymptotic performances would still remain the same with respect to
parameters $n,D,g$).

Our approach is based on propagating the source message to locally and online elected leaders of nearby boxes
first, and then to the remaining nodes in those boxes.
The main challenge in this process is the lack of knowledge about neighbor location.
We solve it through a cascade of diluted transmissions, each initiated by already elected
nearby temporary leaders who try to eliminate other leaders in close proximity.
This size of this proximity is exponentially increasing in the cascade of these elimination processes, so that
at the end only a few nearby leaders in reasonably large distance (to assure a long ``hop'' of
the source message) survive and are used as relays.
In case the network granularity is unknown,
strongly selective families of specifically selected parameters are used in elimination process.
Subtle technical issues need to be solved to avoid simultaneous transmissions of many
nodes in one region, as it not only disturbs local receivers but may also interfere with faraway
transmissions (recall that in case of weak devices, it is not possible to guarantee
such property, as there is no scalable broadcasting algorithm).
Once all local leaders possess the source message, it is simultaneously propagated to their
neighbors in boxes in a sequence of diluted transmissions.

\comment{%
In the former model, we develop a randomized broadcasting
algorithm with time complexity $O(D+\log (1/\delta))$, where $D$ is the eccentricity of
the communication graph, and $\delta$ is the maximal error probability.
This analysis is complemented by the results of simulations on uniform and social networks,
which compare favorably with the performance of exponential backoff protocol.
In the latter model, we give a solution with time complexity $O((D+\log (1/\delta))\log n)$.
Moreover, a deterministic algorithm in the latter model is presented, working in
only slightly larger time $O(D\log^2 n)$.
All these results hold for model parameter $\alpha>2$;
for $\alpha=2$ the randomized solutions are slower by factor $\log^2 n$ and the deterministic one
becomes slower as well.
Additionally, we provide a variation of the deterministic algorithm that works in
$O(D)$ rounds for $\alpha>2$ and in $O(D\log n)$ rounds for $\alpha=2$,
assuming stations know the (constant) network granularity.
Some technical details are deferred to the full version of this paper,
due to the page limitation.
}

\vspace*{-1ex}
\section{Model, Notation and Technical Preliminaries}

We consider a wireless network of $n$ {\em stations}, also called {\em nodes},
deployed into a two dimensional Euclidean space.
Stations communicate by using a (single-frequency) wireless channel.
They have unique integer IDs in set $[\cI]$,
where the size of the domain $\cI$ is bounded by some polynomial in $n$.
(We use the notation $[i,j]=\{k\in\NAT\,|\,i\leq k\leq j\}$ and $[i]=[1,i]$,
for any two positive integers $i,j$.)
Stations are located on the plane with {\em Euclidean metric} $\dist(\cdot,\cdot)$,
and each station knows its Euclidean coordinates.
Each station $v$ has its {\em fixed transmission power} $P_v$, which is a positive real number;
whenever station $v$ chooses to transmit a message, it uses
its full transmission power $P_v$.
In this work we consider a uniform transmission power setting in which $P_v=P$, for
some fixed $P>0$ and every station $v$.
There are three fixed model parameters, related to the physical nature of wireless medium
and devices:
path loss
$\alpha>2$,
threshold $\beta\ge 1$, and ambient noise $\cN\ge 1$.
The $SINR(v,u,\cT)$ ratio, for given stations $u,v$ and a set of (transmitting) stations $\cT$,
is defined as follows:
\begin{equation}\label{e:sinr}
SINR(v,u,\cT)
=
\frac{P_v\dist(v,u)^{-\alpha}}{\cN+\sum_{w\in\cT\setminus\{v\}}P_w\dist(w,u)^{-\alpha}}
\end{equation}
%
In the {\em classical Signal-to-Interference-and-Noise-Ratio (SINR) model},
station $u$ successfully receives a message from station $v$ in a round if
$v\in \cT$, $u\notin \cT$, and

\vspace*{-2ex}
$$SINR(v,u,\cT)\ge\beta \ ,$$

\vspace*{-1ex}
\noindent
where $\cT$ is the set of stations transmitting at that round.

However, in practice the above SINR-based
condition is too simplistic to capture the complexity of the environment,
especially in case of ad hoc networks~\cite{losowy-szum}.
In this work, we consider two enhanced versions of the classical SINR model,
well-established in the literature:
the SINR model with opportunistic links, and
the SINR model with random disturbances.
In the former model, the SINR ratio is used to decide about successful message delivery,
however some links (between faraway nodes) are not taken into account in progress analysis.
In the latter model, each SINR ratio is modified by some random factor.
In this work we consider a general setting with no restrictions on independence of these random
disturbances over nodes, nor their specific distributions.
The only two assumptions made are:
(i) each random factor is in the interval $(1-\eta,1+\eta)$ with probability at least $1-\zeta$,
for some constant parameters $\eta,\zeta\in (0,1)$, and
(ii) for any given pair of nodes, disturbances of the SINR for these two nodes are independent over rounds.

In order to specify the details of broadcasting task and performance analysis in both models,
we first need to introduce the notion of transmission ranges and communication graphs.


\paragraph{Ranges and uniformity.}
The {\em communication range} $r_v$ of a station $v$ is the radius of the ball in which a message transmitted
by the station is heard, provided no other station transmits at the same time. That is, $r_v$ is the largest
value such that $SINR(v,u,\cT)\geq \beta$, provided $\cT=\{v\}$ and $d(v,u)=r_v$.
%
As mentioned before, in this paper we consider uniform networks, i.e.,
when ranges
of all stations are equal (to some constant $P$).
Thus, $r_v=r$ for $r=\left(\frac{P}{\beta\cN}\right)^{1/\alpha}$ and each station $v$.
\tj{For simplicity of presentation and wlog,} we assume that $r=1$,
which implies that $P=\beta \cN$.
The assumption that $r=1$ can be dropped without changing
asymptotic formulas for presented algorithms and lower bounds.

\remove{
Under these assumptions, $r_v=r=(1+\eps)^{-1/\alpha}$
for each station $v$.
The {\em range area} of a station with range $r$
located at the point $(x,y)$ is defined as the circle with radius $r$.
}



\paragraph{Communication graph and graph notation.}
The {\em communication graph} $G(V,E)$
of a given network
consists of all network nodes and edges $(v,u)$ such that $d(v,u)\leq (1-\eps)r=1-\eps$,
\tj{where $\eps<1$} is
a fixed model parameter.
The meaning of the communication graph is as follows: even though the idealistic communication
range is $r$, it may be reached only in a very unrealistic case of single transmission in the whole
network. In practice, however, many nodes located in different parts of the network often
transmit simultaneously, and therefore it is reasonable to assume that we may
hope for a slightly smaller range to be achieved.
The communication graph envisions the network of such ``reasonable reachability''.
Observe that the communication graph is symmetric for uniform networks, which are considered
in this paper.
By a {\em neighborhood} of a node $u$ we mean the set (and positions) of all
neighbors of $u$ in the communication graph $G(V,E)$
of the underlying network, i.e., the set $\{w\,|\, (w,u)\in E\}$.
The {\em graph distance} from $v$ to $w$ is equal to the length of a shortest path from $v$ to $w$
in the communication graph, where the length of a path is equal to the number of its edges.
The {\em eccentricity} of a node
is the maximum graph
distance from this node to all other nodes
(note that the eccentricity is of order of the diameter if the communication
graph is symmetric --- this is also the case in this work).

\tj{
We say that a station $v$ transmits {\em $c$-successfully} in a round
$t$ if $v$ transmits a message in round $t$ and this message is received by
each station $u$ in Euclidean distance from $v$ smaller or equal to $c$. }
We say that node $v$ transmits {\em successfully} to node $u$ in a round $t$ if $v$ transmits
a message in round $t$ and $u$ receives this message.
A station $v$ transmits {\em successfully} in round $t$ if it transmits
successfully to each of its neighbors in the communication graph.
%
%

\paragraph{Synchronization and rounds.}
It is assumed that algorithms work synchronously in time slots, also called {\em rounds}:
each station can act
either as a sender or as a receiver during a round.
We do not assume global clock ticking; algorithm could easily synchronize their rounds by updating
round counter and passing it along the network with messages.

\paragraph{Collision detection.}
We consider the model {\em without collision detection}, that is,
if a station $u$ does not receive a message in a round $t$, it has no information
whether any other station was transmitting in that round,
and no information about the received signal,
e.g., no information about the value of $SINR(v,u,\m{T})$, for any station $v$,
where $\m{T}$ is the set of transmitting stations in round $t$.

\paragraph{Broadcast problem and performance parameters.}
In the broadcast problem studied in this work, there is one distinguished node, called a {\em source},
which initially holds a piece of information, also called a {\em source message} or a
{\em broadcast message}.
The goal is to disseminate this message to all other nodes by sending messages along the network.
Detail performance specification depends on the considered model.

\noindent
{\em Broadcast in the SINR model with opportunistic links:}
The complexity measure is the worst-case time to accomplish the broadcast task,
taken over all networks with specified parameters that have their communication graphs
{\em for fixed model parameter $\ep\in (0,1)$} connected.

\noindent
{\em Broadcast in the SINR model with random disturbances:}
The complexity measure is the worst-case time to accomplish the broadcast task,
taken over all networks with specified parameters that have their communication graphs
{\em defined for $\ep=1$} connected.
Note that in this model $\ep$ is not used as a model parameter, but only with the fixed value $1$
to specify the range of admissible networks. Intuitively, the admissible networks in this case are those
connected according to ``average links'', that is, links
for which the expected transmission ranges (i.e., based on expected modified SINR) are taken into account.
Observe that the broadcasting time is a random variable, even for deterministic algorithms,
due to random disturbances incurred by the model.

Time, also called the {\em round complexity}, denotes here the number of communication rounds in
the execution of a protocol: from the round when the source is activated
with its broadcast message till the broadcast task is accomplished (and each station is aware that
its activity in the algorithm is finished).
For the sake of complexity formulas, we consider the following parameters:
$n$, \dk{$\cI$,} $D$, and $g$, where
$n$ is the number of nodes,
\dk{$[\cI]$} is the \dk{domain}
of IDs,
$D$ is the eccentricity of the source,
and $g$ is the granularity of the network, defined as $r$ times the
inverse of the minimum distance between any two stations (c.f.,~\cite{EmekGKPPS09}).

\paragraph{Messages and initialization of stations other than source.}
We assume that a single message sent in the execution of any algorithm
can carry the broadcast message and at most polynomial in the
size of the network $n$ number of control bits in the size of the network
\tj{(however, our randomized algorithms need only logarithmic number of control bits)}.
For simplicity of analysis, we assume that every message sent during the execution
of our broadcast protocols contains the broadcast message; in practice, further optimization
of a message content could be done in order to reduce the total number of transmitted bits in real executions.
A station other than the source starts executing the broadcasting protocol
after the first successful receipt of the broadcast message; we call it
a {\em non-spontaneous wake-up model}, to distinguish from other possible settings,
not considered in this work,
where stations could be allowed to do some pre-processing
(including sending/receiving messages) prior receiving
the broadcast message for the first time.
We say that a station that received the broadcast message is {\em informed}.

\paragraph{Knowledge of stations.}
Each station knows its own ID, location coordinates, and parameters $n$,
\dk{$\cI$.}
(However, in randomized solutions, IDs can be chosen randomly from
the polynomial range such that each ID is unique with high probability.)
Some subroutines use the granularity $g$ as a parameter, though
our main algorithms can use these subroutines without being aware
of the actual granularity of the input network.
We consider two settings: one with local knowledge of density, in which each station knows also the
number of other stations in its close proximity (dependent on the $\eps$ parameter) and the other
when no extra knowledge is assumed.

\remove{
depending on the algorithm, other general network parameters such as:
diameter $D$ of the imposed communication graph,
or granularity of the network $g$, defined as the inverse of the smallest
distance between any pair of stations.\footnote{%
In many cases, the assumption about the knowledge of $D,g$ can be dropped,
by running parallel threads for different ranges of values of these parameters and implementing an additional coordination mechanism between the threads.}
}

\vspace*{-1ex}
\subsection{Grids}

Throughout the paper, we use notation $\NAT$ for the set of natural numbers,
$\NAT_+$ for the set $\NAT\setminus\{0\}$, and $\INT$
for the set of integers.
Given a parameter $c>0$, we define 
a partition of the $2$-dimensional space
into square boxes of size $c\times c$ by the grid $G_c$, in such a way that:
all boxes are aligned with the coordinate axes,
point $(0,0)$ is a grid point,
each box includes its left side without the top
endpoint and its bottom side without the right endpoint and
does not include its right and top sides.
We say that $(i,j)$ are the coordinates
of the box with its bottom left corner located at $(c\cdot i, c\cdot j)$,
for $i,j\in \INT$. A box with coordinates
$(i,j)\in\INT^2$ is denoted $C_c(i,j)$ or $C(i,j)$ when the side of a grid
is clear from the context.

Let $\eps$ be the parameter defining the communication graph.
Then, $z=(1-\eps)r/\sqrt{2}$ is the largest value such that
the each two stations located in the same box of the grid $G_{z}$
are connected in the communication graph.
Let $\eps'=\eps/3$, $r'=(1-\eps')r=1-\eps'$ and $\gamma'=r'/\sqrt{2}$.
We call $G_{\gamma'}$ the {\em pivotal grid}, borrowing terminology from
radio networks research~\cite{DessmarkP07}.
%

Boxes $C(i,j)$ and $C'(i',j')$ are {\em adjacent} if
$|i-i'|\leq 1$ and $|j-j'|\leq 1$ (see Figure~\ref{fig:adjacent}).
For a station $v$ located in position $(x,y)$ on the plane we define its {\em grid
coordinates} $G_c(v)$ with respect to the grid $G_c$ as the pair of integers $(i,j)$ such that the point $(x,y)$ is located
in the box $C_c(i,j)$ of the grid $G_c$ (i.e., $ic\leq x< (i+1)c$ and
$jc\leq y<(j+1)c$).
\dk{The distance between two different boxes is the maximum Euclidean distance between
any two points of these boxes;
the distance between a box and itself is $0$.}

\remove{ 
A (general) {\em broadcast schedule} $\mathcal{S}$ of length $T$
wrt  $N\in\NAT$ is a mapping
from $[N]$ to binary sequences of length $T$.
A station
with identifier $v\in[N]$ {\em follows}
the schedule $\m{S}$ of length $T$ in a fixed period of time consisting of $T$ rounds,
when
$v$ transmits a message in round $t$ of that period iff
the 
position $t\mod T$ of
$\m{S}(v)$ is equal to $1$.

A {\em geometric broadcast schedule} $\mathcal{S}$ of length $T$
with parameters $N,\delta\in\NAT$, $(N,\delta)$-gbs for short, is a mapping
from $[N]\times [0,\delta-1]^2$ to binary sequences of length $T$.
Let $v\in[N]$ be a station whose grid coordinates
with respect to
the grid $G_c$ are equal to $G_c(v)=(i,j)$.
We say that $v$ {\em follows}
$(N,\delta)$-gbs $\m{S}$ 
for the grid $G_c$
in a fixed period of time, 
when $v$ transmits a message in round $t$ of that period iff
the $t$th position of
$\m{S}(v,i\mod \delta,j\mod\delta)$ is equal to $1$.
}  

\tj{
A set of stations $A$ on the plane is {\em $d$-diluted} wrt $G_c$, for $d\in\NAT\setminus\{0\}$, if
for any two stations $v_1,v_2\in A$ with grid coordinates $(i_1,j_1)$ and $(i_2,j_2)$, respectively,
the relationships $(|i_1-i_2|\mod d)=0$ and $(|j_1-j_2|\mod d)=0$ hold.}

\remove{ 
Let $\m{S}$ be a general broadcast schedule wrt $N$ of length $T$,
let $c>0$ and $\delta>0$, $\delta\in\NAT$.
A $\delta$-dilution of a  $\m{S}$ 
is defined as a $(N,\delta)$-gbs $\m{S}'$ such that the bit $(t-1)\delta^2+a\delta+b$
of $\m{S}'(v,a,b)$ is equal to $1$ iff the bit $t$ of $\m{S}(v)$
is equal to $1$. That is, each round $t$ of $\m{S}$ is
partitioned
into $\delta^2$ rounds
of $\m{S}'$, indexed by pairs $(a,b)\in [0,\delta-1]^2$, such that a station
with grid coordinates $(i,j)$ in $G_c$ is
allowed to send messages only in rounds with index $(i\mod\delta,j\mod\delta)$,
provided schedule $\m{S}$ admits a transmission in its (original) round $t$.
%
%
} 

\begin{figure}
\begin{center}
\epsfig{file=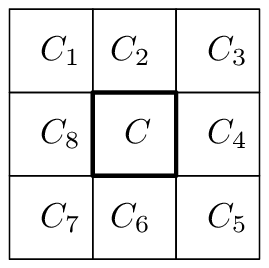, scale=0.7}
\end{center}
\vspace*{-3ex}
\caption{%
The boxes
$C_1,\ldots,C_8$ are adjacent to $C$.}
\label{fig:adjacent}
\vspace*{-2ex}
\end{figure}%
\remove{ 
Observe that, since ranges of stations are equal to the length
of diagonal of boxes of the pivotal grid, a box $C(i,j)$ can have at most
$20$ neighbors (see Figure~\ref{fig:adjacent}).
We define the set $\DIR\subset[-2,2]^2$ such  that $(d_1,d_2)\in\DIR$ iff
it is possible that boxes with coordinates $(i,j)$ and $(i+d_1,j+d_2)$
can be neighbors.
Given $(i,j)\in\INT^2$ and $(d_1,d_2)\in\DIR$, we say that the box $C(i+d_1,j+d_2)$
is {\em located in direction} $(d_1,d_2)$ from the box $C(i,j)$.
} 

\comment{%
\tj{For a family $\m{F}=(F_0,\ldots,F_{k-1})$ of subsets of $[N]$, an {\em execution}
of $\m{F}$ on a set of stations $V$ is a protocol in which $v\in V$ transmits in
rund $i$ iff $v\in F_{i\mod k}$.
}
A family $S$ of subsets of $[N]$ is a {\em $(N,k)$-ssf (strongly-selective family)}
if, for every non empty subset $Z$ of $[N]$ such that
$|Z|\leq k$ and for every element $z\in Z$, there is a set $S_i$ in $S$ such that
$S_i\cap Z=\{z\}$.
It is known that there exists $(N,k)$-ssf of size $O(k^2\log N)$ for every $k\leq N$,
c.f.,~\cite{ClementiMS01}.
Let $k=(2d+1)^2$, let $S$ be a $(N,k)$-ssf,
$s=|S|=O(\log N)$,
$N=poly(n)$.
The sets $S_0,\ldots,S_{s-1}$ of the family $S$ define a broadcast schedule $S'$ in such a way
that station $v$ transmits in round $t$ iff $v\in S_{t\mod s}$. (That is, $S'_{t\mod s}[v]=1$ iff
$v\in S_{t\mod s}$.)
}














\vspace*{-2ex}
\section{Leader Election in Boxes}
\label{s:leader}

The main goal of this paper is to develop
two deterministic algorithms:
%
one depending on the knowledge of network
granularity, and one general algorithm which
does not need such knowledge.
%
\comment{
The latter general algorithm works
merely for $\alpha>2$. 
(i.e., it is not applicable for $\alpha=2$).
}
The key ingredient of both protocols is a leader election sub-routine.
We consider leader election problem defined as follows.
Given $x\leq (1-\lambda)/\sqrt{2}$, for $0<\lambda<1$,
and a set of ``active'' stations $V$, the goal is to choose
a leader in each box of the grid $G_x$ containing at least
one element of $V$.
%
In this section we design
two algorithms for the defined leader election problem, and 
in the next section
we will show how to apply them to obtain scalable deterministic distributed broadcasting protocols.
%
\dk{In every deterministic algorithm, we assume that each message carries all values stored in its sender.}


\vspace*{-1ex}
\subsection{Granularity-dependent leader election}


%
%
%

%
Let DilutedTransmit($V,x,d$) be the following procedure, consisting of $d^2$
communication rounds:
\begin{algorithm}[H]
	\caption{DilutedTransmit($V,x,d$)}
	\label{alg:diluted}
	\begin{algorithmic}[1]
    \For{each $a,b\in[0,d-1]^2$}
        \State $A\gets \{v\in V\,|\, G_x(v) \equiv (a,b) \mod d\}$
        \State All elements of $A$ transmit a message
    \EndFor
    \end{algorithmic}
\end{algorithm}

Below are two useful properties of DilutedTransmit; see Appendix for the proof of Proposition~\ref{prop:diluted:transmit}.
We say that a function $d_{\alpha}:\NAT\to\NAT$ is {\em flat} 
if
\iffull
\begin{equation}
d_{\alpha}(n)=\left\{
\begin{array}{rcl}
O(1) & \mbox{ for } & \alpha>2\\
O(\log n) & \mbox{ for } & \alpha=2
\end{array}
\right.
\end{equation}
\else
$d_{\alpha}(n)=O(1)$ for $\alpha>2$. 
\fi
\begin{proposition}\labell{prop:diluted:transmit}
Let $V$ be a set of at most $n$ stations such that there is at most one station in each
box of $G_x$ and $x\leq (1-\lambda)/\sqrt{2}$ for $0<\lambda<1$. Then, there exists a flat function
$d_\alpha(n)$ such that each element of $V$ transmits
$(2\sqrt{2}x)$-successfully during
DilutedTransmit($V,x,\sqrt{d_\alpha(n)}$).
\end{proposition}


We say that a box $C$ of the grid $G_x$ has a {\em leader} from set $A$ if there is one station
$v\in A$ located in $C$
with status {\em leader} and all stations from $A$ located in $C$ know which station it is.

\begin{proposition}\labell{prop:lead2}
Assume that $A$ is a set of leaders in some boxes of the grid $G_x$, where $x\leq\frac{1-\lambda}{2\sqrt{2}}$, and
each station knows whether it belongs to $A$. Then,
it is possible to choose the leader of each box of $G_{2x}$ containing at least one element of $A$
in 
$O(\frac{d_{\alpha}(n)}{\lambda})$ rounds, where $d_{\alpha}$ is a flat
function.
%
\end{proposition}

\begin{proof}
Note that each box of $G_{2x}$ consists of four boxes of $G_x$. Let us fix some labeling of this four
boxes by the numbers $\{1,2,3,4\}$, the same in each box of $G_{2x}$.
Now, assign to each 
station from $A$
the label $l\in\{1,2,3,4\}$ corresponding
to its position in the box of $G_{2x}$ containing it.
We ``elect'' leaders in $G_{2x}$ in four phases
$F_1,\ldots,F_4$. Phase $F_i$ is just the execution of
DilutedTransmit($A,x,d$) for $d=(d_\alpha(n)/\lambda)^{1/\alpha}$ and
$A$ equal to the
set of leaders with label $i$ (see Proposition~\ref{prop:lead1} in the Appendix).
%
%
Therefore, each leader from $A$ can hear messages of all other (at most)
three leaders located in the same box of $G_{2x}$. Then, for a box $C$ of $G_{2x}$, the leader with the
smallest label (if any) among leaders of the four sub-boxes of $C$ becomes the leader of $C$.
Finally, complexity bound stated in the proposition follows directly from Proposition~\ref{prop:lead1}
in the Appendix and inequality $\alpha>2$.
\end{proof}

\paragraph{Algorithms LeadIncrease and GranLeaderElection.}
\tj{Let LeadIncrease($A,x,\lambda$) denote a procedure, which,
given leaders of boxes of $G_x$,
chooses leaders of boxes of $G_{2x}$ in $O(\frac{d_{\alpha}(n)}{\lambda})$ rounds.
Such a procedure exists by Proposition~\ref{prop:lead2}.
Repeating this procedure sufficiently many times for different sets of input parameters,
we obtain the following granularity-dependent 
leader election algorithm.
}

\iffull
Assume that granularity of a network is
equal to $g$.
Given $z\leq (1-\lambda)/\sqrt{2}$, our goal is to choose leaders of boxes of $G_z$. To this aim
we choose the smallest $i$ such that $\frac{z}{2^i}\leq \frac1{g}$. Then, each station
can be recognized as the leader of its box of $G_{z}$. Next, the protocol LeadIncrease
is executed repeatedly until the leaders of boxes of $G_z$ are obtained.
\else
\fi


\begin{algorithm}[H]
	\caption{GranLeaderElection($V,g,z$)}
	\label{alg:gran}
	\begin{algorithmic}[1]
    \State $x\gets\max\{\frac{z}{2^i}\,|\, i\in\NAT,\frac{z}{2^i}\leq \frac1{g}\}$
    \State $A\gets V$ \Comment{Each station is a leader of its box of $G_x$}
    \State $d\gets d_\alpha(n)$
    \While{$x\leq z/2$}
        \State $\lambda\gets (1-2\sqrt{2}x)$
        \State LeadIncrease($A,x,\lambda$)
        \State $A\gets $ leaders of boxes of $G_{2x}$
        \State $x\gets 2x$
    \EndWhile
    \end{algorithmic}
\end{algorithm}
\iffull
Finally, we summarize properties of Algorithm GranLeaderElection in the following
proposition.
\else
Let $d_\alpha(n)$ be a flat function from Proposition~\ref{prop:lead2}.
\begin{theorem}\labell{t:leader:gran}
Given $z<1/\sqrt{2}$, the
algorithm GranLeaderElection$(V,g,z)$ chooses a leader in each box of
the grid $G_z$ containing at least one element of $V$ in
\tj{$O((1/\lambda+\log (gz))d_{\alpha}(n))$ rounds, where $\lambda=1-\sqrt{2}z$ and
granularity of a network is at most $g$.}
\end{theorem}

\begin{proof}
Correctness \tj{of GranLeaderElection} follows from properties of LeadIncrease and choice of parameters
(see Proposition~\ref{prop:lead2}). Proposition~\ref{prop:lead2} and the choice of $x$
\tj{in line 1 of GranLeaderElection directly
imply the bound $O(\frac{d_{\alpha}(n)\log (gz)}{\lambda})$. }
However, all but
the last execution of LeadIncrease is called with $\lambda\geq 1/2$ which gives
the result.
\end{proof}

\vspace*{-2ex}
\subsection{General leader election}

In the following, we describe leader election algorithm 
that chooses leaders of boxes of the grid $G_z$ in
$O(\log^2 n/\lambda^2)$ rounds,
provided $z < 1/\sqrt{2}$, $\alpha>2$
and $\lambda=1-\sqrt{2}z$.

For a family $\m{F}=(F_0,\ldots,F_{k-1})$ of subsets of $[\cI]$, an {\em execution}
of $\m{F}$ on a set of stations $V$ is a protocol
in which $v\in V$
\tj{transmits in round $i\in[0,t-1]$ iff $v\in F_{i\mod k}$.}
A family $S$ of subsets of $[\cI]$ is a {\em $(\cI,k)$-ssf (strongly-selective family)}
if, for every non empty subset $Z$ of $[\cI]$ such that
$|Z|\leq k$ and for every element $z\in Z$, there is a set $S_i$ in $S$ such that
$S_i\cap Z=\{z\}$.
It is known that there exists $(\cI,k)$-ssf of size $O(k^2\log \cI)$ for every $k\leq \cI$,
c.f.,~\cite{ClementiMS01}.

\tj{In the algorithm choosing leaders of boxes of $G_z$ for $z=(1-\lambda)/\sqrt{2}$,
we use a $(\cI,k)$-ssf family $S$ of size
$s=O(\log \cI)$, where $k$ is a constant depending on $\lambda$ and on $\alpha>2$.
We will execute $S$ on various sets of stations.
}
The set $X_v$, for a given
execution of $S$ and station $v$, is defined as the set of IDs of stations belonging to $\boxx_z(v)$
such that $v$ can hear them during that execution of $S$.

\def\leaddesc{
The leader election algorithm GenLeaderElection$(V,z)$
chooses leaders from $V$ in boxes of $G_z$.
It consists of two stages.
The first stage gradually eliminates
the set of candidates \tj{for leaders (simultaneously in all boxes)}
in consecutive executions of a strongly-selective family $S$.
It is implemented as a for-loop. We call this stage {\em Elimination}.

Let {\em block} $l$ of Elimination stage denote the executions of family $S$ for $i=l$.
Each ``eliminated'' station $v$ has
assigned the value $ph(v)$, which is equal to the number of the block in which it is eliminated.
\tj{Let $V(l)=\{v\,|\, ph(v)>l\}$ and $V_C(l)=\{v\,|\, ph(v)>l\mbox{ and } \boxx_z(v)=C\}$, for $l\in\NAT$}
and $C$ being a box of grid $G_z$. The key properties of sets $V_C(l)$ are:
$|V_C(l+1)|\leq |V_C(l)|/2$ for every box $C$ and $l\in\NAT$,
and the granularity of $V_C(l_C^{\star})$ is smaller than $n/z$
for every box $C$ and $l_C^{\star}$ being the largest $l\in\NAT$ such that $V_C(l)$
is not empty.
%
Thus, in particular, $V_C(l)=\emptyset$ for each $l\geq \log n$ \tj{and each box $C$ of $G_z$}.
%
Motivated by the above observations, the algorithm in its second stage chooses the leader of
each box $C$ by applying --- simultaneously in each box --- the granularity-dependent
leader election algorithm GranLeaderElection
on $V_C(\log n)$, $V_C(\log n-1)$, $V_C(\log n-2)$ and so on, until each box has its leader elected.
The second stage of the algorithm is called {\em Selection}.
}

We provide the
pseudo-code of the leader election algorithm in Algorithm~\ref{alg:leader}, and 
then its correctness and complexity analysis will proceed (some technical details are deferred to the appendix). 
All references to boxes in the algorithm regard boxes of $G_z$.

\begin{algorithm}[t!]
	\caption{GenLeaderElection($V,z$)}
	\label{alg:leader}
	\begin{algorithmic}[1]
    \State For each $v\in V$: $cand(v)\gets true$; 
    \For{$i=1,\ldots,\log n+1$}\Comment{Elimination}
        \For{$j,k\in[0,2]$}
            \State Execute $S$ twice on the set: 
\Comment{$S$ is $(\cI,d)$-ssf of length $O(d^2\log \cI)$, $d$ large enough~\cite{ClementiMS01}} 
            \State \ \ $\{w\in V\,|\,cand(w)=true, w\in C_z(j',k')
			\mbox{ such that } (j',k')\equiv(j,k)\mod 2\}$;
            \State Each $w\in V$ determines and stores $X_w$ during
			the first execution of $S$ and 
		\State \ \ \ $X_v$, for each
			$v\in X_w$, during the second execution of $S$, where
		 \State \ \ \  $X_u$ is the set of nodes from box of $u$ 
                heard by $u$ during execution of $S$ on $V$; 
            \For{each $v\in V$}
                \State $u\gets \min(X_v)$
                \If{$X_v=\emptyset$ or $v>\min(X_u\cup\{u\})$}
                    \State $cand(v)\gets false$; $ph(v)\gets i$
                \EndIf
            \EndFor
        \EndFor
    \EndFor

    \State For each $v\in V$: $state(v)\gets active$ \Comment{Selection}
    \For{$i=\log n,(\log n)-1,\ldots,2,1$}
        \State $A_i\gets \{v\in V\,|\, ph(v)=i, state(v)=active\}$
        \State $V_i\gets$ GranLeaderElection($A_i,n/z,z$)
        \State $\lambda\gets 1-\sqrt{2}z$ \tj{\Comment{$V_i$ -- new leaders}}
        \State For each $v\in V_i$: $state(v)\gets leader$ 
        \State DilutedTransmit($V_i,z,d$) for $d=(d_\alpha(n)/\lambda)^{1/\alpha}$ 
        \State For each $v\in V$ which can hear $u\in\boxx(v)$ 
		\tj{during DilutedTransmit($V_i,z,d$): 
$state(v)\gets passive$}
    \EndFor
    \end{algorithmic}
\end{algorithm}

{\leaddesc}

\begin{theorem}
\labell{t:leader:general} 
Algorithm GenLeaderElection$(V,z)$ chooses a leader in each box of $G_z$ containing
at least one element of $V$ in
$O(\log^2 n)$ rounds,
provided $\alpha>2$ and $\lambda=1-\sqrt{2}z>0$ \tj{are constant}.
\end{theorem}

\vspace*{-2ex}
\section{Broadcasting Algorithms}
\label{s:broadcast}


\newcommand{\asleep}{asleep}
\newcommand{\activ}{active}
\newcommand{\passive}{passive}

We first describe a generic algorithm DetBroadcast, which uses leader election protocol
in boxes of grid $G_{z}$ for $z=\eps'/\sqrt{2}$, where $\eps'=\eps/2$, as a subroutine
(recall that $\eps$ is the constant defining the communication graph).
The performance of the algorithm is estimated in two variants:
the first 
in which network granularity is known (and GranLeaderElection is applied),
and the second which uses
GenLeaderElection and does not depend on network granularity.

Let $\gamma'=(1-\eps')/(2\sqrt{2})$.
At the beginning of the algorithm, all stations except of the source $s$ are in the
state {\asleep} (states of stations in broadcasting algorithm are independent of their
states during their calls to leader election subroutines).
In the first round of DetBroadcast, the source sends a message to all
stations in its range area; these stations become {\activ}) while the source
changes its state to {\passive}.
Then, the algorithm works in stages $1,2,3,\ldots$, where the stage $i$ consists of:
\begin{itemize}
\vspace*{-0.3ex}
\item
{\em one execution} of the
leader election procedure GenLeaderElection($V_i,z$)
or \\ GranLeaderElection($V_i,g,z$), where $z=\eps'/\sqrt{2}$ and $V_i$ is the set of station in state
{\activ} at the beginning of the stage,
followed by
\vspace*{-0.3ex}
\item
{\em $(\gamma'/\eps')^2$ applications} of
DilutedTransmit($V'_{i,a,b},\gamma',d$) indexed by pairs $(a,b)\in [0,d-1]^2$, where
$V'_i$ are the leaders of boxes of $G_z$ chosen from $V_i$,
$V'_{i,a,b}$ are elements of $V'_i$ with grid coordinates \tj{(with respect to $G_z$}
equal to $(a,b)$ modulo $d$
\tj{and $d=d_\alpha(n)/(2\gamma')$}.
\end{itemize}
The goal of these \tj{``diluted''} applications of DilutedTransmit is that
leaders of boxes of $G_{z}$ (acting as leaders of boxes of $G_{\gamma'/2}$)
send messages
to all neighbors (in the communication
graph) of all stations
from their boxes of $G_{z}$. In order to achieve this goal, it is sufficient that leaders
transmit $(1-\eps')$-successfully.
At the end of stage $i$, all stations in state {\em active} become {\em passive} and
all stations in state asleep, which received the broadcast message during stage $i$, change
state to {\em active}.


Below we present a pseudo-code of a stage of the broadcasting algorithm
DetBroadcast.
\begin{algorithm}[H]
	\caption{StageOfBroadcast
\Comment{a single stage of algorithm DetBroadcast}}
	\label{alg:stage}
	\begin{algorithmic}[1]
    \State $\eps'\gets \eps/2$; $\gamma'\gets 1-\eps'$; $z\gets \eps'/\sqrt{2}$
    \State $l\gets \lceil \gamma'/\eps'\rceil$
    \State $V\gets$ stations in state {\em \activ}
    \State Run leader election sub-routine: either GenLeaderElection($V,z$) or GranLeaderElection($V,g,z$)
    \State $V'\gets$ leaders chosen during the leader election in line 4
    \For{each $(a,b)\in[0,l-1]^2$}
        \State $V'_{a,b}=\{v\in V'\,|\, G_z(v)\equiv (a,b)\mod l\}$
        \tj{\State $d\gets (d_{\alpha}(n)/\eps')^{1/\alpha}$ \Comment{$d_{\alpha}$ from Prop.~\ref{prop:lead1}}
        \State DilutedTransmit($V'_{i,a,b},(1-\eps')/(2\sqrt{2}),d$)}
     \EndFor
     \State for each $v$: if $state(v)={\activ}$: $state(v)\gets \passive$
     \State for each $v$: if $state(v)=\asleep$ and $v$ received the broadcast message: $state(v)\gets \activ$
    \end{algorithmic}
\end{algorithm}


\begin{lemma}
\label{l:detbroadcast}
Algorithm DetBroadcast accomplishes broadcasting in $O(D)$ stages,
provided the leader election sub-routine in line 4 of StageOfBroadcast correctly elects leaders in all boxes
of grid $G_x$.
\end{lemma}

\begin{proof}
We first formulate an essential fact for correctness of our broadcasting algorithm, which
easily follows from the definition of a reachability graph.

\begin{fact}\labell{f:boxall}
Let $\eps'=\eps/2$ for $\eps<1$.
If a station $v$ from a box $C$ of a grid $G_x$ for $x\leq \eps/(2\sqrt{2})$ transmits a message
$(1-\eps')$-successfully then its message is received by all neighbors (in the
reachability graph) of all stations located in $C$.
\end{fact}

Each station $v$ which receives the broadcast message for the
first time at stage $j$, changes its state from {\asleep} to {\activ} at the end of stage $j$.
Then, at the end of stage $j+1$, such station $v$ changes its state from {\activ} to {\passive}.
In each stage, only (and exactly) the stations in state {\activ} take part as transmitters
in leader election and DilutedTransmit.
Fact~\ref{f:boxall} guarantees that, if a station $v$ is in the state {\em \activ} during
stage $j$, then all its neighbors in the reachability graph receive the broadcast message 
during DilutedTransmit (in line 9 of StageOfBroadcast)
\tj{in stage $j$. 
Therefore all neighbors of $v$ (in the reachability graph) are in the state {\activ} in stage $j+1$ or earlier. This implies that broadcasting
is finished after $O(D)$ applications of Algorithm~\ref{alg:stage}.}
\end{proof}

Let DetGenBroadcast and DetGranBroadcast denote the broadcasting algorithm using 
GenLeaderElection and GranLeaderElection, respectively, in line 4 of StageOfBroadcast.
Time performances of these leader election protocols, together with Lemma~\ref{l:detbroadcast}, 
imply the following results.

\begin{theorem}\labell{t:broadgen}
Algorithm DetGenBroadcast accomplishes broadcast in $O(D\log^2 n)$ rounds, provided
$\alpha>2$, \tj{$\eps<1/2$ are constant.}
\end{theorem}

\begin{proof}
The result holds by applying Lemma~\ref{l:detbroadcast} together with Theorem~\ref{t:leader:general}
regarding performance of algorithm GenLeaderElection used for leader election in line 4 of 
StageOfBroadcast,
and by using the facts that the size of $S$
is $O(\log \cI)$ and that $\lambda'/\eps'$ is constant. 
\end{proof}

\begin{theorem}\labell{t:broadgran}
Algorithm DetGranBroadcast accomplishes broadcast in 
$O(D(1/\eps^3+\log g)d_{\alpha}(n))=O(D\log g)$ rounds, for
constant parameters 
$\alpha>2$
and~$\eps<1/2$.
\end{theorem}

\begin{proof}
Complexity of GranLeaderElection for $z=\eps'/\sqrt{2}$ is $O(d_\alpha(n)\log g)$, since $1-\sqrt{2}z$ is larger
than $1/3$, see Theorem~\ref{t:leader:gran}. Then, $l=\Theta(1/\eps)$, and $d=\Theta((d_\alpha(n)/\eps)^{1/\alpha})$. Therefore,
the for-loop works in $O((1/\eps)^3 d_{\alpha}(n))$ rounds. Combining this with Lemma~\ref{l:detbroadcast}
yields the theorem.
\end{proof}

\vspace*{-3ex}
\section{Model with Randomly Disturbed SINR}
\labell{s:random-SINR}

In this section we show simple modifications of original procedures and 
algorithms from Sections~\ref{s:leader} and~\ref{s:broadcast},
and argue that their performance in the model with randomly disturbed SINR is
bigger by factor $O(\log_{1/\zeta} n)$ than the performance of the original versions
analyzed in the model with opportunistic links in Sections~\ref{s:leader}
and~\ref{s:broadcast}. For simplicity, whenever we discuss original algorithms,
they are understood to be analyzed in the opportunistic links model,
while with respect to the modified algorithms, we assume that they are
studied in the randomly disturbed SINR model.

We emulate each round of the original algorithms by 
$\tau=\Theta(\log_{1/\zeta} n)$ consecutive rounds, and we call them a {\em phase}.
That is,
each round of the original algorithms, 
which we call {\em original round}, 
is replaced by a single phase containing $\tau$ rounds.
Each node transmitting in an original round 
transmits in all $\tau$ rounds of the corresponding phase.
However, the local computation done after 
receiving the signal from the wireless medium in the original round
is done only once in the corresponding phase---after receiving the signal
from the wireless medium in the final round $\tau$ of the phase.

Note that there are more possibilities of receiving messages in a phase, comparing with
the corresponding original round, due to random disturbances of SINR ratios.
Therefore, in all phases of the modified protocols, 
each node ignores all messages successfully received from nodes of distance bigger than $(1-\ep)r$
from it, where $\ep\in (0,1)$ is defined in such a way that the randomly modified SINR ratio of two nodes
of distance at most $(1-\ep)r$ is above the threshold $\beta$ with probability at least $\zeta$
(note that $\ep$ depends on all parameters $\alpha,\beta,\eta,\zeta$).
In fact, in the analysis of the original algorithms in Sections~\ref{s:leader} and~\ref{s:broadcast}
we measured progress only in terms of such opportunistic transmissions between nodes
of distance at most $(1-\ep)r$ from each other; therefore if
we explicitly ignore any other (faraway) transmissions in the modified algorithms,
we receive the same feedback (from the wireless channel) in phases as in the corresponding original 
rounds, with high probability (whp). (The exact probability is at least $1-n^{-c}$, where
$c>1$ is a constant depending on the constant hidden in ``$\Theta$'' notation in the definition of $\tau$.)
This is because of three facts: 
(a) ignoring messages from nodes of distance bigger than $(1-\ep)r$ allows to focus on the same neighbors
as in the progress analysis of the original protocols;
(b) the probability that a node does not receive a message from another node of distance at most 
$(1-\ep)r$ from it in any round of a given phase, provided it received it in the corresponding original round,
is at most $\zeta^\tau$, and
(c) sufficiently large parameter $\tau$ makes the probability small enough inverse of polynomial in $n$
in order to be able to use union bounds of events over all nodes and rounds when transforming the original
analysis.

In the Appendix we show that local computation done by original algorithms 
can be directly transformed into the modified versions of these algorithms as well;
this is because they are designed to assure a high level of knowledge consistency,
and because the messages received in the original executions are also received in the
executions of modified algorithms, whp (as showed above). 
Thus, enhancing the results in Theorems~\ref{t:broadgen} and~\ref{t:broadgran} 
by additional factor $\tau=O(\log n)$ coming from simulating each original round
by a phase of $\tau$ rounds,
we obtain the following result.

\comment{
Our goal is 
Therefore we show how to adapt local computations used in rounds of original algorithms
in Sections~\ref{s:leader} and~\ref{s:broadcast} to corresponding phases, 
in such a way that the measurable results of these local computations in original rounds 
(proved in the model with opportunistic links)
are the same as in corresponding phases (proved in the model with random disturbances),
with high probability. This will be enough to argue, after applying a union bound
inequality over all original rounds in an execution of original algorithm,
that the executions of the modified algorithms mimic the corresponding executions
of original algorithms, and thus the broadcasting task is accomplished  by the modified algorithms
with performance growing by factor $O(\log_{???} n)$, with high probability
in the model with random disturbances. Note that although the modified algorithms are also
deterministic, their performance is given only with high probability, as the executions
contain random factors (i.e., disturbances of SINR) incurred by the considered model.
}

\def\TRANSFORM{
Below we analyze each type of the original rounds (of algorithms in Sections~\ref{s:leader} and~\ref{s:broadcast}), and argue that they can be transformed to the model with random disturbances
of SINR when applying the general transformation described in Section~\ref{s:random-SINR}.

\paragraph{Algorithm~\ref{alg:diluted}: DilutedTransmit.}
This procedure does not contain any specific local computation, only transmission
pattern (each round of which is simply copied $\tau$ times, as specified above in the definition
of a phase).

\paragraph{Procedure LeadIncrease.}
See the proof of Proposition~\ref{prop:lead2} regarding specification of the original procedure 
LeadIncrease($A,x,\lambda$).
In this procedure, leaders of smaller boxes of size $x$ run procedure DilutedTransmit, and at the end 
the smallest of other at most 3 leaders within the larger box of size $2x$ 
(containing 4 smaller boxes in total) elects itself as the leader of this box,
while the others know it. 
In the model with randomly disturbed SINR, by the observation above,
with high probability each leader of smaller box receives ids of all other (at most 3) leaders
of small boxes within the larger box, and therefore all (at most 4) of them select
the same leader among themselves using the smallest-id rule as in 
the original procedure LeadIncrease.

\paragraph{Algorithm~\ref{alg:gran}: GranLeaderElection.}
This algorithm simply iterates modified procedure LeadIncrease, which gives the same
result as the original LeadIncrease, whp,
with respect to exponentially growing
boxes, and no new local computation rules are used.

\paragraph{Algorithm~\ref{alg:leader}: GenLeaderElection.}
It contains two parts: elimination and election. In the elimination part,
the local computation proceeds with checking conditions (lines 5, and 11)
and updating variables (lines 10 and 12). A straightforward inductive argument
over the number of runs of the internal part of the loop, lines 4-13, 
guarantees that, with high probability, sets $X_u$ are the same, and thus the values of 
variables $cand$ and $ph$ are the same as in the corresponding original round.
Based on them, transmissions are scheduled, which are again the same, whp.
In the selection part, only the already computed variables $ph$ are taken into account
(and we argued that they are the same as in the original execution, whp), and
based on them the modified procedures GranLeaderElection and DilutedTransmit are executed,
about which we also argued that they yield the same results as the original ones.

\paragraph{Generic algorithm DetBroadcast.}
It is sufficient to examine the specification of a single stage of the original algorithm,
given in Algorithm~\ref{alg:stage}, StageOfBroadcast.
In a stage, only a leader election is run once (either GranLeaderElection or GenLeaderElection),
which, as we showed, works the same in its original and modified versions, whp.
Additionally, common knowledge (known parameters, location) is used, and
procedure DilutedTransmit is executed a few times (again, its modified
version works the same as the original one, whp). 
We conclude that both implementations of the generic DetBroadcast algorithms ---
one based on GranLeaderElection and the other based on GenLeaderElection ---
work the same in their modified forms as they worked in the original forms, whp.

}

\begin{theorem}
The modified version of algorithm DetGenBroadcast accomplishes broadcast in $O(D\log^3 n)$ rounds, 
and algorithm DetGranBroadcast accomplishes broadcast in 
$O(D\log g\log n)$ 
rounds, with high probability.
\end{theorem}

\comment{

\section{Conclusions and Future Work}

In this work we showed the first provably well-scalable deterministic distributed solutions for the broadcast problem in {\em any} uniform wireless network under the physical model
based on SINR with random deviations.
Several novel techniques developed in this work for the purpose of leader election and broadcast,
could be extended for finding scalable solutions to other communication problems.

}


\bibliographystyle{abbrv}



\clearpage

\appendix

\begin{center}
{\LARGE\bf Appendix}
\end{center}



\section{Useful Properties of Diluted Transmissions and the Proof of Proposition~\ref{prop:diluted:transmit}}

First, we define some useful notation and prove one technical proposition.

Let $I_1=[i_1,j_1)$, $I_2=[i_2,j_2)$ be segments
on a coordinate axes, whose endpoints belong to the grid
$G_x$.
The {\em max-distance} between $I_1$ and $I_2$ with respect go $G_x$ is zero when $I_1\cap I_2\neq\emptyset$,
and it is equal to $\min(|i_1-j_2|/x, |i_2-j_1|/x)$ otherwise. Given two rectangles
$R_1$, $R_2$, whose vertices belong to $G_x$, the max-distance $\distM(R_1,R_2)$ between $R_1$
and $R_2$ is equal to the maximum of the max-distances between projections
of $R_1$ and $R_2$ on the axes defining the first and the second dimension in the Euclidean
space.

\begin{proposition}\labell{prop:lead1}
For each $\alpha>2$
and $\lambda<1$,
there exists a flat function $d_{\alpha}(n)$
such that
the following property holds.
Assume that a set of $n$ stations $A$ is $d$-diluted wrt the grid $G_x$, where $x\leq (1-\lambda)/(2\sqrt{2})$ 
and $d\geq \left(d_{\alpha}(n)/\lambda\right)^{1/\alpha}$.
Moreover, $\min_{u,v\in A}(\dist(u,v)\geq\sqrt{2}x)$ (i.e., in particular,
at most one station from $A$ is located in each box of $G_x$). Then,
if all stations from $A$ transmit simultaneously, each of them
is $2\sqrt{2}x$-successful. Thus, in particular, each station from
a box $C$ of $G_x$ can transmit its message
to all its neighbors located in $C$ and in boxes $C'$ of $G_x$ which are adjacent to $C$.
\end{proposition}
%
\begin{proof}
Consider any station $u$ in distance smaller or equal to $2\sqrt{2}x<3x$ to a station $v\in A$.
Then, the signal from $v$ received by $u$ is at least
$$\frac{P}{(2\sqrt{2}x)^{\alpha}}.$$
Now, we would like to derive an upper bound on interferences caused by stations in $A\setminus\{v\}$
at $u$.
Let $C$ be a box of $G_x$ which contains $v$.
The fact that $A$ is $d$-diluted
wrt $G_x$ implies that the number of boxes containing elements of $A$
which are in max-distance
$id$ from $C$ is at most $8(i+1)$ (see Figure~\ref{fig:gran}).
Moreover, no box in distance $j$ from $C$ such
that ($j\mod d\neq 0$) contains elements of $A$.
\begin{figure}
\begin{center}
\epsfig{file=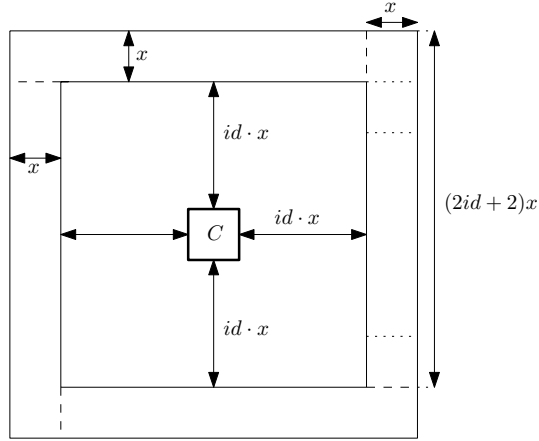, scale=0.8}
\end{center}
\caption{Boxes in distance $id$ from $C$ form a frame partitioned into four
rectangles of size $x\times (2id+2)x$. Each of these rectangles contain at most $i+1$
boxes such that any two of them are in max-distance at least $d$.}
\label{fig:gran}
\end{figure}%
Finally, for a station $v\in C$ and a station $w\in C'$ such that $\distM(C,C')=j$,
the inequality $\dist(v,w)\geq jx$ is satisfied.
Note that our goal is ({\em not}) to evaluate interferences (at $v\in C$, but) at any station
$u$ such that $\dist(u,v)\leq \frac{2\sqrt{2}x}c<3x$. Therefore, $u\in C'$ such that $\distM(C,C')<3$, where
$C'$ is a box of $G_x$.
For a fixed $d$, the total noise and interferences $I$ caused by
all elements of $A\setminus\{v\}$ at any location in $C$ is at most
$$\cN+\sum_{i=1}^{n}8(i+1)\cdot\frac{P}{(idx)^{\alpha}}$$
Since $\distM(C,C')<3$ and $\distM$ satisfies the triangle inequality,
$\distM(C',C'')\geq \distM(C,C'')-3$ for each $C''$, where $u\in C'$. Therefore,
if $d>3$, the noise at $u$ is at most
$$\cN+\sum_{i=1}^{n}8(i+1)\cdot\frac{P}{(i\bar{d}x)^{\alpha}}$$
where $\bar{d}=d-3$.
%
Furthermore,
$$I\leq \cN+ 8\cdot\left(\frac{P}{\bar{d}x}\right)^{\alpha}\cdot\sum_{i=0}^{n}(i+1)^{1-\alpha}\leq \cN+8e_{\alpha}(n)\left(\frac{P}{\bar{d}x}\right)^{\alpha}$$
where $e_{\alpha}(n)=\sum_{i=1}^{n}i^{1-\alpha}=1+\zeta(\alpha-1)$, $\zeta$ is the Riemann zeta function. 
So, 
the signal from $v$ is received at $u$ if the following
inequality is satisfied
\begin{equation}\label{eq:signal}
\beta\left( \cN+8e_\alpha(n)\left(\frac{P}{\bar{d}x}\right)^{\alpha}\right)\leq \left(\frac{P}{2\sqrt{2}x}\right)^{\alpha}
\end{equation}
which can be shown equivalent to
$$\bar{d}^\alpha\geq \frac{8\cdot(2\sqrt{2})^\alpha \beta}{(1-(2\sqrt{2}x)^\alpha)}\cdot e_\alpha(n)$$
using the assumption $P=\beta\cN$. Since $(2\sqrt{2}x)\leq 1-\lambda$ and therefore
$\frac{1}{1-(2\sqrt{2}x)^\alpha}\leq \frac1{\lambda}$,
it is sufficient that
$$\bar{d}_\alpha\geq \frac{8(2\sqrt{2})^\alpha\beta}{\lambda}\cdot e_{\alpha}(n)$$
and, since $\alpha>2$,
that
$$\bar{d}\geq \left(\frac{d_\alpha(n)}{\lambda}\right)^{1/\alpha}$$
where $d_\alpha(n)=2\sqrt{2}\cdot (8\beta)^{1/\alpha}\cdot (e_{\alpha}(n))^{1/\alpha}$.
%
%
\end{proof}

\noindent
If the smallest distance between elements of $V$ is larger than or equal to
$\sqrt{2}x$ (for $x\leq (1-\lambda)/(2\sqrt{2})$) then, according to Proposition~\ref{prop:lead1}, each station $v\in V$
transmits successfully its message to all its neighbors located in boxes of $G_x$
adjacent to $\boxx_x(v)$ during execution of DilutedTransmit($V,x,d$) for $d=(d_\alpha(n)/\lambda)^{1/\alpha}$.

\comment{

\begin{proposition}\labell{prop:diluted:transmit}
Let $V$ be a set of $\leq n$ stations such that there is at most one station in each
box of $G_x$ and $x\leq (1-\delta)/\sqrt{2}$. Then, there exists a flat function
$d_\alpha(n)$ such that each element of $V$ transmits
$2\sqrt{2}x$-successfully during
DilutedTransmit($V,x,\sqrt{d_\alpha(n)}$).
\end{proposition}

}

Using the above Proposition~\ref{prop:lead1}, we can prove Proposition~\ref{prop:diluted:transmit}
from the main part of the paper.
\\

\noindent\textbf{Proof of Proposition~\ref{prop:diluted:transmit}:}
\tj{This proposition is a simple corollary from Proposition~\ref{prop:lead1},}
as DilutedTransmit($V,x,d$) splits $V$ into $d$-diluted subsets.
\qed


\comment{

\vspace*{2ex}
\noindent\textbf{Proof of Theorem~\ref{t:leader:gran}:} 
Correctness \tj{of GranLeaderElection} follows from properties of LeadIncrease and choice of parameters
(see Proposition~\ref{prop:lead2}). Proposition~\ref{prop:lead2} and the choice of $x$
\tj{in line 1. of GranLeaderElection directly
imply the bound $O(\frac{d_{\alpha}(n)\log (gz)}{\lambda})$. }
However, all but
the last execution of LeadIncrease is called with $\lambda\geq 1/2$ which gives
the result.
\qed 

} 

\section{Deferred Details from the Analysis of GenLeaderElection: Proof of Theorem~\ref{t:leader:general}}


First, we analyze some properties of communication in the SINR model which eventually will
justify application of strongly selective families \tj{in GenLeaderElection}.

\begin{proposition}
\labell{prop:for:sel}
For each $\alpha>2$ and $\lambda<1$, there exists a constant $d$, which depends
only on $\alpha$,
satisfying the following property.
Let $W$ be a set of stations such that
$\min_{u,v\in W}\{\dist(u,v)\}=\sqrt{2}x\leq (1-\lambda)\}$
(i.e., in particular,
there is at most one station from $W$ in each box
of $G_x$). 
Let $u,v$ be the pair of closest stations, i.e., $\dist(u,v)=\sqrt{2}x$
and let $u\in C$, where $C$ is a box of $G_x$.
If $u$ is transmitting in a round $t$ and no other station
in any box $C'$ of $G_x$ in the max-distance at most $d/\lambda^{1/(\alpha-2)}$ from $C$ is transmitting at that round,
then $v$ can hear the message from $u$ at round $t$.
\end{proposition}
\begin{proof}
Recall that, according to our assumptions, $P=\beta\cN$ and $\sqrt{2}x\leq 1-\lambda$.
The power of signal from $u$ received by $v$ is then
$$S=\frac{P}{(\sqrt{2}x)^\alpha}=\frac{\beta\cN}{(\sqrt{2}x)^\alpha}$$
Assuming that no station
in any box $C'$ in the max-distance at most $d$ from $C$ is transmitting,
the amount of interference and noise at $v$ is at most
$$
I\leq\cN+\sum_{i=d}^{\infty}8(i+1)\cdot\frac{1}{(ix)^{\alpha}}
=
\cN+\frac{8}{x^{\alpha}}\cdot e_d
\ ,
$$
where $e_d=\sum_{i=d+1}^{\infty}i^{1-\alpha}$. 
Thus, it is sufficient to show that $S\geq \beta I$, i.e.,
$$\frac{\beta\cN}{(\sqrt{2}x)^\alpha}\geq \beta(\cN+\frac{8}{x^\alpha}\cdot e_d)$$
which is equivalent to
\tj{$$e_d\leq \frac{\cN(1-(\sqrt{2}x)^\alpha)}{8\cdot 2^{\alpha/2)}}$$}
Since $\sqrt{2}x\leq 1-\lambda$ (and therefore $1-(\sqrt{2}x)^\alpha\geq$ $1-(1-\lambda)^\alpha\geq$ $\lambda$, the above inequality is satisfied for
each $e_d$ such that
$$e_d\leq \frac{\cN\lambda}{8\cdot 2^{1/(2\alpha)}}$$
which holds for sufficiently large $d$, because of convergence of $\sum_{i=1}^{\infty}1/i^{\alpha}$ for $\alpha>2$.

\tj{More precisely, $e_d=\sum_{i=d+1}^{\infty}i^{1-\alpha}\leq\frac{d^{2-\alpha}}{\alpha-2}$
due to the fact that $\sum_{i=d+1}^{\infty}i^{1-\alpha}\leq \int_d^{\infty}i^{1-\alpha}$.
Thus, it is sufficient that the following inequality is satisfied
$$d\geq \left(\frac{8\cdot 2^{\alpha/2}}{\cN\lambda(\alpha-2)}\right)^{\frac{1}{\alpha-2}}.$$
}
\end{proof}

\tj{\begin{corollary}\labell{cor:selector}
For each $\alpha>2$ and $\lambda<1/3$, there exists a constant $k$ 
satisfying the following property.
Let $W$ be a set of stations such that
$\min_{u,v\in W}\{\dist(u,v)\}=\sqrt{2}x$ and let $\dist(u,v)=\sqrt{2}x$
for some $u,v\in W$.
Then, $v$ can hear the message from $u$ during an execution of a $(\cI,k)$-ssf on $W$.
\end{corollary}
\begin{proof}
Let $d'=\frac{d}{\lambda^{\alpha-2}}$, where $d$ is the constant from Proposition~\ref{prop:for:sel}. 
Let $u,v\in W$ be as stated in the corollary and let $l=(2d'+1)^2$. Conditions of the corollary
imply that there is at most one station from $W$ in each box of $G_x$. Let $S$ be
$(\cI,l)$-ssf. Thus, during execution of $S$ for $|S|$ rounds, there exists a round $t$ in which
$u$ send a message and no other station in any box at the max-distance at most $d'$ from
$\boxx_x(u)$ sends a message. Proposition~\ref{prop:for:sel} implies that $v$ can hear
$u$ in round $t$.
\end{proof}
}
While Corollary~\ref{cor:selector} says that a pair of closest station can exchange messages
during execution of $(\cI,k)$-ssf, the following proposition generalizes this result by guaranteeing
that, simultaneously, the closest pair of stations in each box $C$ can exchange messages, provided
there are close enough stations in $C$.
\begin{proposition}\labell{prop:closer}
For each $\alpha>2$ and $\lambda<1/3$, there exists a constant $k$ 
satisfying the following property.
Let $z\leq (1-\lambda)/\sqrt{2}$, let $W$ be a $d$-diluted for $d\geq 3$ wrt $G_z$  set of stations and let
$C$ be a box of $G_z$. Moreover, let
$\min_{u,v\in W, \boxx_z(u)=\boxx_z(v)=C}\{\dist(u,v)\}=\sqrt{2}x\leq z/n$ and $\dist(u,v)=\sqrt{2}x$
for some $u,v\in W$ such that $\boxx_z(u)=\boxx_z(v)=C$.
Then, $v$ can hear the message from $u$ during an execution of a $(\cI,k)$-ssf on $W$.
\end{proposition}
\begin{proof}
Let $u,v$ and $x$ be as specified in the proposition and let $C=\boxx_z(u)=\boxx_z(v)$ be
a box of $G_z$.
Let $S$ be a $(\cI,k)$-ssf.
If all stations
from $W$ are located in $C$, then the claim follows directly from Corollary~\ref{cor:selector}.
So, let $W'$ be the set of all elements of $W$ which are {\em not} located in $C$.
Let us (conceptually) ``move'' all stations from $W'$ to boxes adjacent to $C$,
preserving the invariant that $\min_{u,v\in W, \boxx(u)=\boxx(v)=C}\{\dist(u,v)\}=x$. Note that such
a movement is possible, since there are at most $n$ stations in $W'$ and the side of a box of the
grid $G_z$ is larger
than $1/2$. Since $W$ is $3$-diluted, the distance from $w\in C$ to any station $w'\in W'$ before
movement of $w'$ is larger than the distance from $w$ to $w'$ after movement.
Let $W''$ define $W$ with new locations of stations (after movements).
Therefore, if $u$ can
hear $v$ in the execution of $S$ on $W''$ (i.e., after movements of stations), it can hear $v$
in the execution of $S$ on $W$ (i.e., with original placements of stations).
However, the fact that $u$ can hear $v$ on $W''$ follows directly from \tj{Corollary~\ref{cor:selector} by the fact that
$\min_{u,v\in W''}\{\dist(u,v)\}=x$}.
\end{proof}

Next, we concentrate on properties of Algorithm~\ref{alg:leader}.
Recall the notation: $V(l)=\{v\,|\, ph(v)>l\}$ and $V_C(l)=\{v\,|\, ph(v)>l\mbox{ and } \boxx(v)=C\}$, for $l\in\NAT$
and $C$ being a box of $G_z$.

\begin{proposition}\labell{prop:lead:empty}
Let $C$ be a box of $G_z$ for $z\leq (1-\lambda)/\sqrt{2}$, $\delta<1/3$ and $l\in\NAT$.
Then,
\begin{enumerate}
\item[(i)]
$|V_C(l+1)|\leq |V_C(l)|/2$;
\item[(ii)]
If $V_C(l+1)$ is empty, then the smallest distance between elements of $V_C(l)$ is larger
than $z/n$.
\end{enumerate}
\end{proposition}
\begin{proof}
Observe that our algorithm implicitly builds matchings
in the graph whose nodes are $V_C(l)$ and an edge connects such $u$ and $v$ that
$u$ can hear $v$ and $v$ can hear $u$ during an execution of $S$.
In other words, $(u,v)\in E$ when $x\in X_v$ and $v\in X_u$. Moreover,
a pair $(u,v)$ belongs to our matching when the following conditions
are satisfied:
 \begin{itemize}
 \item
 $u=\min(X_v)$
 \item
 $v=\min(X_u\cup\{u\})$
 \end{itemize}
and therefore also $v<u$. As one can see, each station $w$ can belong
to (at most) one such a pair which proves this set of pairs
forms a matching in $V_C(l)$ indeed.
Note that the station $v\in V_C(l)$ belongs to $V_C(l+1)$ only if
it is the smaller element of a pair belonging to our matching.
%
Therefore, the inequality $|V_C(l+1)|\leq |V_C(l)|$ holds. This yields
conclusion (i) of the lemma.

As for conclusion (ii), assume that $V_C(l)$ is not empty. Observe that $V_C(l+1)$ is
not empty if there exist $v,u\in V_C(l)$ such that $v$ can hear $u$
and $u$ can hear $v$ during execution of $S$.
(Indeed, $v\in V_C(l+1)$ for the smallest $v\in V_C(l)$ such that $v$ can hear $u$
and $u$ can hear $v$ for some $u\in V_C(l)$.) However, such $v$ and $u$ exist if the smallest
distance between elements of $V_C(l)$ is at most $\frac{z}{n}$ by Proposition~\ref{prop:closer}.
\end{proof}

Finally, we are ready to prove Theorem~\ref{t:leader:general}.
\\

\noindent\textbf{Proof of Theorem~\ref{t:leader:general}:}
Time complexity $O(\log n\log \cI)$ follows immediately from the bounds on the size of selectors
and complexity of GranLeaderElection.

Proposition~\ref{prop:lead:empty}(i) implies that $V_C(l)=\emptyset$ for each box $C$ and $l>\log n$.
(In other words, $ph(v)\leq\log n$ for each $v\in V$.)
For a box $C$ of $G_z$, let $l_{\star}=\max_l\{V_C(l)\neq\emptyset\}$.
By Proposition~\ref{prop:lead:empty}(ii), the smallest distance between stations of $V_C(l_{\star})$ is at least $z/n$.
In other words the smallest distance between stations of
$\{v\in V\,|\, ph(v)=l_{\star}, state(v)=active\}$ is $\geq z/n$, where $l_{\star}$ is the largest
number $l$ such that $ph(v)=l$ for some $v\in V$.

Let us focus on a box $C$ of $G_z$ which contains at least one station from $V$.
Selection stage 
tries to choose the leader of $C$ among
$V_C(\log n), V_C(\log n-1), \ldots$. Moreover, when the leader is elected, all stations
from $C$ are switched off (i.e., their state is set to passive which implies that they
do not attend further GranLeaderElection executions).
Since $l_{\star}=\max_l(V_C(l)\neq\emptyset)\leq\log n$
and the smallest distance between elements of $V_C(l_{\star})$ is $\geq z/n$, each execution of GranLeaderElection
is applied on a set of stations with the smallest distance between stations $\geq z/n$ 
\tj{implying granularity $\Theta(n/z)$}, and therefore the leader
in each box $C$ containing (at least one) element of $V$ is chosen.
\qed

\section{Transforming Algorithms to the SINR Model with Random Disturbances}

\TRANSFORM

\end{document}
